\documentclass[sigconf]{acmartmod}

\setcopyright{none}
\renewcommand\footnotetextcopyrightpermission[1]{}

\usepackage{balance}

\makeatletter
\renewcommand\@formatdoi[1]{\ignorespaces}
\makeatother
\renewcommand\footnotetextcopyrightpermission[1]{}
\usepackage{booktabs}
\usepackage{url}
\usepackage{graphicx}
\usepackage{xcolor}
\usepackage{xargs}
\usepackage{xspace}

\usepackage[linesnumbered,lined,noend]{algorithm2e}
\SetKwComment{kc}{$\triangleright$~}{}
\SetKwComment{tcp}{$\triangleright$~}{}

\SetCommentSty{mycommfont}
\DontPrintSemicolon

\DeclareMathOperator*{\argmax}{arg\,max}
\newcommand{\abs}[1]{\ensuremath{\left|{#1}\right|}}
\newcommand{\tup}[1]{\ensuremath{\left\langle{#1}\right\rangle}}
\newcounter{obs}
\counterwithin{obs}{section}
\newenvironment{obs}[1][]{\refstepcounter{obs}\par\medskip
    \textsc{Observation~\theobs. \textit{#1}} \rmfamily}{\medskip}

\def\STI/{\textsf{ST-Index}}
\def\Traversal/{\textsf{Traversal}}
\def\Order/{\textsf{Order}}
\def\FS/{\textsf{FromScratch}}
\def\SE/{\textsf{SingleEdge}}
\def\BA/{\textsf{Batch}}
\def\K/{\ensuremath{\mathcal{K}}}
\def\C/{\ensuremath{\mathcal{C}}}
\def\H/{\ensuremath{\mathcal{H}}}

\title{Batch Dynamic Algorithm to Find $k$-Cores and Hierarchies}
\author{Kasimir Gabert}
\additionalaffiliation{%
    \institution{Sandia National Laboratories}
    \city{Albuquerque}
    \state{New Mexico}
    \country{USA}
}
\affiliation{%
    \institution{Georgia Institute of Technology}
    \city{Atlanta}
    \state{Georgia}
    \country{USA}
}
\email{kasimir@gatech.edu}
\author{Ali P{\i}nar}
\affiliation{%
    \institution{Sandia National Laboratories}
    \city{Livermore}
    \state{California}
    \country{USA}
}
\email{apinar@sandia.gov}
\author{\"{U}mit V. \c{C}ataly\"{u}rek}
\additionalaffiliation{%
    \institution{Amazon Web Services. This publication describes work performed at the Georgia Institute of Technology and is not associated with Amazon}
}
\affiliation{%
    \institution{Georgia Institute of Technology}
    \city{Atlanta}
    \state{Georgia}
    \country{USA}
}
\email{umit@gatech.edu}

\begin{document}

\begin{abstract}
    Finding $k$-cores in graphs is a valuable and effective strategy for
    extracting dense regions of otherwise sparse graphs.  We focus on the
    important problem of maintaining cores on rapidly changing dynamic
    graphs, where batches of edge changes need to be processed quickly.
    Prior batch core algorithms have only addressed half the problem of
    maintaining cores, the problem of maintaining a core decomposition. This finds
    vertices that are dense, but not regions; it misses connectivity.  To
    address this, we bring an efficient index from community search into
    the core domain, the Shell Tree Index.  We develop a novel dynamic
    batch algorithm to maintain it that improves efficiency over
    processing edge-by-edge.  We implement our algorithm and
    experimentally show that with it core queries can be returned
    on rapidly changing graphs quickly enough for interactive applications.
    For 1 million edge batches, on many graphs we run over $100\times$ faster than
    processing edge-by-edge while remaining under re-computing from scratch.
\end{abstract}

\maketitle

\section{Introduction}

An important problem in graph analysis is finding locally dense regions in
globally sparse graphs.  In this work we consider the problem of finding
$k$-cores~\cite{seidman1983network,matula1983smallest}, which are maximal
connected subgraphs with minimum degree at least $k$.  This problem has
seen significant attention given its efficiency~\cite{matula1983smallest}
and usefulness across many domains~\cite{kumar2000web,alvarez2005k,
hagmann2008mapping,van2011rich,garcia2017ranking,filho2018hierarchical,
kong2019k}.

Many practically important graphs from web data, social networks, and
related fields are both large and continuously changing.
The problem of
maintaining core \emph{decompositions} on graphs has been well
studied~\cite{li2013efficient,sariyuce2013streaming,zhang2017fast,
zhang2019unboundedness}.  Existing approaches run in linear time in the
size of the graph, which is theoretically
optimal~\cite{zhang2019unboundedness}, and on many real-world graphs they
maintain decompositions within milliseconds after edge changes.  So, is
the problem solved?

Unfortunately, these approaches only address half of the problem of
returning a $k$-core\cite{sariyuce2016fast}.  $k$-cores are originally
defined as \emph{connected} subgraphs~\cite{seidman1983network}.  All of
the application examples referenced above rely on or use connectivity.
A core decomposition, on the other hand, provides \emph{coreness} values
for every vertex: that is, the largest value such that a vertex is in
a $k$-core, but not in a $(k+1)$-core.  Prior approaches have either
ignored connectivity (which provides limited, but some insight
e.g.,~\cite{kitsak2010identification}) or left the final step of finding
components as a separate process.  The main tool to address computing connectivity on cores, or a \emph{core hierarchy},
has been independently proposed several
times~\cite{barbieri2015efficient,sariyuce2016fast,fang2017effective,
fang2020effective} in different contexts, and concurrently developed in~\cite{lin2021hierarchical}.  We introduce this index in the
most basic setting, designed for $k$-cores on simple undirected graphs, and
we call it the Shell Tree Index (\STI/).  This index supports queries to
extract the cores a vertex is in along with the full core
hierarchy of a graph.

\begin{figure}[tbp]
    \centering
    \includegraphics{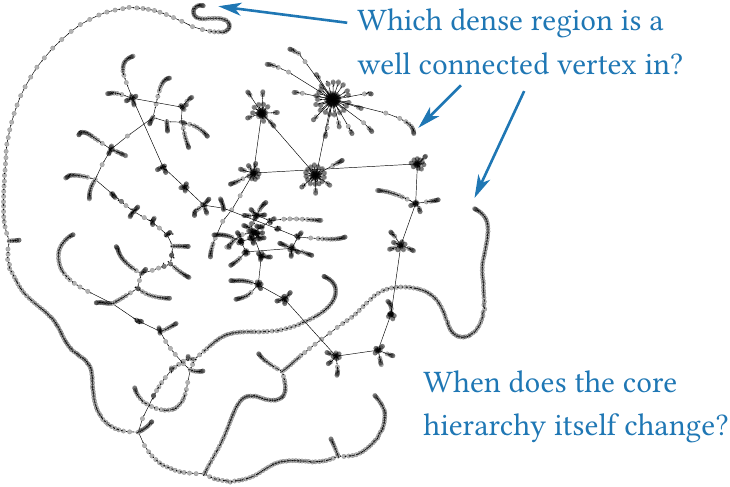}
    \caption{\label{fig:lj}
    The core hierarchy for the LiveJournal social network graph. Tracking
    dense regions behavior over time is important for understanding
    structural changes, and extracting the vertices within a dense region
    is important for almost all known $k$-core applications.
    }
\end{figure}

\paragraph{Example Problem}

Consider the problem of managing a social network.  First, given a user,
we wish to recommend friends to them that are well connected \emph{in
their part} of the graph: this is a vertex and coreness query.  Second, we
want to detect structural changes, for example sybil
attacks~\cite{douceur2002sybil} from new fake accounts: this is
a hierarchy query.  Figure~\ref{fig:lj} shows the core hierarchy of the
LiveJournal graph~\cite{yang2015defining} and how far apart different
dense regions are.  For both query examples, we want results in tens of
milliseconds to either prepare a webpage or mitigate an emerging attack.

In this example scenario, a state-of-the-art core decomposition system is
put in place, which provides coreness updates quickly after graph changes.
The two goals above require information about \emph{specific cores}.  If
certain vertices achieve higher coreness values, this does not inform
whether a new region is created.  Furthermore, unless there is only one
dense region, it will not enable useful recommendations.  Instead, we need
systems and algorithms that can quickly and effectively return \emph{cores
themselves} along with their full hierarchies.

\paragraph{Approach}

The \STI/ builds on the laminar nature of cores.
For $k' < k$, every $k$-core is contained within some $k'$-core, naturally forming a tree.
Each node in this tree corresponds to the \emph{shell} of the core, that is
vertices which are not in any higher core.
Coupled with a reverse map, a core can
be efficiently returned by traversing the subtree staying below the
desired $k$ value.
The core hierarchy is the tree.

We build the tree by first identifying regions of the graph where the
cores are the same, known as subcores, and then forming a directed acyclic
graph (DAG) with each subcore as a node.
Starting from the highest $k$ values, we process nodes in the DAG upwards, merging and moving them to form a tree.

The only known prior maintenance approach, operating for attributed graphs and used as part of solving community search, is given
in~\cite{fang2017effective}.
We first port this maintenance
approach to the case of $k$-cores on standard graphs and use that as our edge-by-edge baseline.
Given an edge
change, it maintains the \STI/ by either merging or splitting nodes on paths to the root.
Concurrent with this work, \cite{lin2021hierarchical} builds on~\cite{fang2017effective}'s approach by batching operations on the tree.

In real-world graphs there is significant variance in the rate of change.
As such, batch dynamic algorithms that can reduce the
total work when operating on batches are
desired~\cite{luo2020batch,dhulipala2020parallel}.  We provide a
batch dynamic algorithm to maintain cores themselves, starting from
core decompositions.
We do this by maintaining the subcore DAG used during construction.
After a batch of changes, we revisit each node in the DAG that was modified and re-compute any subcore changes.
Any DAG changes are then pushed into the tree, temporarily turning the tree back into a DAG.
We then traverse from the sink upwards, correcting the tree.

\paragraph{Contributions}
In addition to bringing the \STI/ from the community search domain into the direct, $k$-core domain, we prove efficiency properties on the \STI/.
Our main contributions are:
\begin{enumerate}
    \item A subcore DAG based \STI/ construction algorithm
    \item A batch dynamic algorithm to maintain \STI/ that
        reduces the work of edge-by-edge updates
    \item An experimental evaluation on real-world graphs that
        show with both our edge-by-edge and batch algorithms, \STI/ is suitable for interactive use
\end{enumerate}

The remainder of this paper is structured as follows.  In
\S~\ref{sec:related} we describe the related work.  In
\S~\ref{sec:preliminaries} we formally describe our model and problem.  In
\S~\ref{sec:sti} we present \STI/.  In \S~\ref{sec:computing-sti} we
provide our algorithm to compute \STI/ from scratch.  In
\S~\ref{sec:maint-sti} we explain how to maintain \STI/ for dynamic graphs
and introduce our batch algorithm.  In \S~\ref{sec:experiments} we
experimentally evaluate our implementations, and in
\S~\ref{sec:conclusion} we conclude.

\section{Related Work} \label{sec:related}

$k$-cores were introduced independently
in~\cite{seidman1983network,matula1983smallest}.
\cite{matula1983smallest} additionally provided a peeling algorithm that uses
bucketing to run in $O(n + m)$.  The main
strategy for computing $k$-cores has remained roughly the same since then:
iteratively peeling the graph, or excluding vertices with too low of
degrees, until all degrees are $k$.

For maintenance, \cite{li2013efficient} and \cite{sariyuce2013streaming}
independently proposed \Traversal/, which limits consideration of vertices around an edge
change if they provably cannot update values. \cite{sariyuce2013streaming}
defines the
notion of subcores and purecores, variants of which are used in all known
maintenance algorithms to limit considered subgraphs. \cite{zhang2017fast}
proposed \Order/, which is the current state-of-the-art and maintains
a \emph{peeling order}, instead of coreness values directly, using an order-statistic treap
and a heap.
Parallel approaches have relied on identifying
a set of vertices that can be independently peeled~\cite{jin2018core,
aksu2014distributed,hua2019faster,aridhi2016distributed}.
\cite{bai2020efficient,zhang2019unboundedness} provide batch algorithms
that reduce work as multiple edges are processed simultaneously.

All of the above focus on computing the \emph{coreness values} for vertices.  In
fact, the lack of focus on connectivity has, in some cases, resulted in
later work redefining cores to not include connectivity
(e.g.,~\cite{malliaros2020core}) which limits their usefulness.

Numerous other targets, similar to cores, have been
proposed~\cite{malliaros2020core}.  \cite{eidsaa2013s,zhou2020core}
develop weighted extensions to cores, \cite{linghu2020global} uses core
concepts to
reinforce connections within networks, \cite{galimberti2020core} proposes
notions of cores for multilayer networks, and \cite{zhang2020exploring}
ensures vertices in core-like regions are also relatively
cohesive given their neighbors.
In cases where the cores are used for downstream
algorithms, returning the actual (connected) vertices is identified as
crucial and algorithms are built to support such
queries~\cite{liu2019efficient}.

Community search~\cite{sozio2010community,cui2014local} is a more general
problem for returning a connected set of vertices in a community based on
a seed set.  The community is commonly defined with a \emph{minimum
degree} measure\cite{fang2020survey}.  In this case, if the query consists
of a single vertex, community search can return exactly a core.  For this
reason, we pull from the field of community search to develop \STI/.
\cite{barbieri2015efficient} proposed the first known shell tree index.
It does not support efficient queries, as it creates additional vertices
for each coreness level that must be addressed.  \cite{sariyuce2016fast}
identifies the same problem that we address---cores require
connectivity---and proposes a shell tree-like index with a static construction in the more general nuclei framework, but leaves out
maintenance. \cite{fang2017effective} operates on attributed graphs and extends
\cite{sariyuce2016fast}'s approach and \cite{barbieri2015efficient}'s
index with incremental and decremental algorithms, but without batch
algorithms.  We port this approach to the problem of cores and use this as our baseline.
Concurrent with this work, \cite{lin2021hierarchical} provides a batch algorithm that is based on \cite{fang2017effective} and batches changes to the tree directly, without the use of a DAG.

\section{Preliminaries} \label{sec:preliminaries}

A graph $G=(V,E)$ is a set of vertices $V$ and set of edges $E$. An edge
$e \in E$ represents the connection between two distinct vertices $u,v \in
V$, $e=\{u,v\}$.
We denote $\abs{V}$ by $n$ and $\abs{E}$ by $m$.

We use $\Gamma(v)$ to represent the neighboring edges of $v \in V$.
The degree of $v \in V$ is $\abs{\Gamma(v)}$.
For directed graphs, $\Gamma^\mathrm{in}$ represents
edges ending at the given vertex and $\Gamma^\mathrm{out}$ represents
edges leaving a vertex. If the graph is ambiguous, we use $\Gamma_G$ for
graph $G$. The neighborhood of a vertex set $S \subseteq V$, $\Gamma(S)$,
represents the vertices and edges connected to $S$, that is it is the
subgraph induced by $S$ and all neighbors of vertices in $S$.

\paragraph{Dynamic Graph Model}
We consider graphs
that are
changing over time, known as dynamic graphs.  An \emph{edge change} is
a tuple $\tup{c, v, e}$ consisting of a direction $c$,  a vertex $v \in
V$, and an edge $e \in E$.  A dynamic graph is then an infinite turnstile
stream of edge changes $\mathcal{S}$, where time is the position in the
stream. At any point in time an undirected graph
$G^t$
can be formed by applying all edge changes until $t$, starting from an
empty graph.

In this model, the timestamp of edges received is not preserved and not used by the algorithm.
An algorithm that does take into consideration timestamps is called a \emph{temporal} algorithm, and can either be dynamic or static.

\begin{definition}\label{def:alg}
    Let $\mathcal{A}$ be a graph algorithm with output $\mathcal{A}(G)$.
    Then $\mathcal{A}_\Delta$ is a dynamic graph algorithm if, for some
    times $t$ and $t'$, with $t < t'$,
    \begin{equation*}
        \mathcal{A}_\Delta\left(G^{t}, \mathcal{A}(G^{t}),
        \mathcal{A}_\Delta^t,
        \mathcal{S}[t, t']\right)
        = \tup{\mathcal{A}(G^{t'}),
        \mathcal{A}_\Delta^{t'}},
    \end{equation*}
    where $\mathcal{A}_\Delta^t$ contains algorithm state at $t$ and
    $\mathcal{S}[t,t']$ are the edge changes in $\mathcal{S}$ from $t+1$ to
    $t'$.
\end{definition}

We call an \emph{incremental algorithm} a dynamic graph algorithm which
can only handle edge insertions and
a \emph{decremental algorithm} one which can only handle edge deletions.
A \emph{batch dynamic} algorithm can handle $t' > t+1$.
Our batch algorithm, described in
Section~\ref{sec:maint-sti}, has an additional state bound by the size of
the graph.

\paragraph{Cores}

We provide a brief background on $k$-cores.

\begin{definition}\label{def:core}
    Let $G$ be a graph
    and $k \in \mathbb{N}$.  A $k$-core
    in $G$ is a set of vertices $V'$ which induce a subgraph $K=(V', E')$
    such that: (1) $V'$ is maximal in $G$; (2) $K$ is connected; and (3)
    the minimum degree is at least $k$, $\min_{v \in V'} \abs{\Gamma_K(v)}
    \geq k$.
\end{definition}

Figure~\ref{fig:example-kcore} shows an example graph and its
cores.
There are two \emph{separate} $k=3$ cores, one with vertices $1$ through $4$ and the other with vertices $7$ through $10$.
If all vertices with less than a degree $3$ are iteratively removed, the remaining graph consists of those two separate connected components.

\begin{figure}
    \centering
    \includegraphics{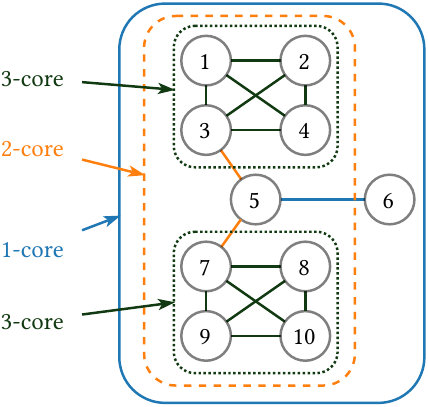}
    \caption{\label{fig:example-kcore}
    An example graph and its cores. Note that there are two separate
    $3$-cores.
    }
\end{figure}

\begin{definition}\label{def:coreness}
    Let $G=(V,E)$ be a graph and $v \in V$.
    The coreness of $v$, denoted $\kappa[v]$, is the value $k$ such that
    $v$ is in a $k$-core
    but not in a $(k+1)$-core.
\end{definition}

\begin{definition}\label{def:degeneracy}
    Let $G=(V,E)$ be a graph.  The $k$-core number of $G$, denoted
    $\rho_G$ and shortened to $\rho$, is given by $\rho = \max_{v \in V}
    \kappa[v]$.
\end{definition}

\paragraph{Problem Statement}

We consider the problem of efficiently supporting core and coreness
queries on a dynamic graph stream. Let $k \in \mathbb{N}$ and $u \in V$.
\begin{itemize}
    \item The coreness query $\K/(u)$ returns $\kappa[u]$.
    \item The core query $\C/(u, k)$ returns the vertices of the $k$-core
        subgraph that contains $u$.
    \item The hierarchy query $\H/$ returns the hierarchical structure of
        the cores as a tree, with the root as the $0$-core
\end{itemize}
Prior work in the context of cores has focused only on supporting $\K/$
queries on dynamic graphs.
Unfortunately, this prevents many of the applications of $k$-cores which
rely on \emph{extracting dense regions} of a graph.

\section{Shell Tree Index} \label{sec:sti}

In this section we present the Shell Tree Index, \STI/, which is able to
efficiently return cores for different vertices: its runtime is asymptotically the size of the result and its space is linear in the number of vertices.  This index
has been independently developed several
times~\cite{barbieri2015efficient,sariyuce2016fast,fang2017effective,
fang2020effective,lin2021hierarchical} in different contexts.
We present the index here for completeness.
We will address how to construct the index in
Section~\ref{sec:computing-sti} and how to maintain it in
Section~\ref{sec:maint-sti}.

$\K/(u)$ queries, or \emph{coreness} queries, can be
efficiently returned using an array of size $n$.  We therefore focus on
$\C/$ and $\H/$ queries.

\begin{figure}
    \centering
    \includegraphics{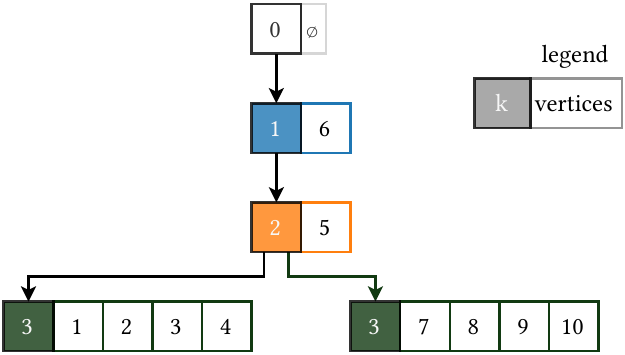}
    \caption{\label{fig:sti}
    The shell tree for the graph shown in Figure~\ref{fig:example-kcore}.
    On the left side are the $k$-shell values, and on the right side are
    the contained vertices.  Each directed edge indicates inclusion of the
    deeper cores.
    }
\end{figure}

\begin{lemma}[\cite{sariyuce2013streaming}]\label{lem:laminar}
    Cores form a laminar family, that is every pair of cores are either disjoint or one is contained in the other.
\end{lemma}
\begin{proof}
    We want to show that for every two cores $K_1$ and $K_2$, $K_1 \cap K_2$ is exactly one of $\emptyset$, $K_1$, or $K_2$.

    Let $K_1$ and $K_2$ be two cores of $G$, with corresponding $k$ values
    $k_1$ and $k_2$.
    Suppose $K_1 \cap K_2 \neq \emptyset$, implying
    $K_1$ is connected to $K_2$.  Note that $k_1 \neq k_2$, otherwise $K_1
    \cup K_2$ is a $k_1$-core, invalidating maximality.  Let $k_1 < k_2$.
    Suppose $\exists v \in K_2$ such that $v \not\in K_1$.  $v$ must be
    connected in $K_2$, and so there exists a path from $K_1$ to $v$ with
    minimum degree at least $k_2$.  Let $K_1'$ be a subgraph that includes
    $K_1$ and the path to $v$.  Then, $K_1'$ is a $k_1$-core and larger
    than $K_1$, invalidating maximality.
\end{proof}

\begin{definition}\label{def:shell}
    Let $G=(V,E)$ be a graph and $K \subseteq V$ a $k$-core in $G$
    for some $k \in \mathbb{N}$.  Then $S$ is a $k$-shell if $S = \{v \in K:
    \kappa[v] = k\}$.
\end{definition}

Note that the shell is \emph{disconnected}, however it is a subset of
a \emph{connected} core. This means that the traditional approach of using
coreness values to compute the shell does not work. We address shell
computation later in Section~\ref{sec:computing-sti}, using
\emph{subcores}.

A \emph{shell tree} $T$ is at the heart of the \STI/.
We call the vertices of $T$ \emph{tree nodes}, to distinguish from the vertices in $G$.
Each node has two additional pieces of data associated with it: a $k$ value and a set of vertices (in $G$).
$T$ is built as follows.
A root node is made with $k=0$ and a vertex set of isolated vertices (those with $\abs{\Gamma_G(v)} = 0$).
Next, nodes are made in $T$ for every $k$-shell.
Its $k$ attribute is set to $k$ corresponding to the shell and its vertex list is set to the vertices in the $k$-shell.
An edge is created in $T$ by linking $k$-shells, following Lemma~\ref{lem:laminar}.
An example shell tree is shown in Figure~\ref{fig:sti}.
The \STI/ consists of $T$ and a map $M$, mapping $v\in V$ to the appropriate node in $T$.

\begin{lemma}\label{lem:st-tree}
    The shell tree is a directed, rooted tree.
\end{lemma}
\begin{proof}
    Suppose a tree node $u$, corresponding to core $K_u$ has two in-edges.
    By definition~\ref{def:shell}, each parent corresponds to a unique
    $k$-shell.  Consider the two corresponding cores, $K_1$ and $K_2$.
    They both include $K_u$, yet are distinct, and so they have
    non-trivial overlap contradicting Lemma~\ref{lem:laminar}.
    The root is defined with $k=0$.
\end{proof}

\begin{lemma}\label{lem:compress}
    The out-degree of a non-root tree node with no corresponding vertices
    in the shell tree can be at most 1.
\end{lemma}
\begin{proof}
    Let the tree node with no corresponding vertices be at level $k > 0$ with out-degree at least 2.
    Then, there are two distinct \emph{cores} at $k+1$ (not necessarily shells), and
    one core at $k$.  The two cores at $k+1$ must be disconnected by construction.

    However, because the tree node has no corresponding vertices, we know that every vertex in the $k$-core is also in a $(k+1)$-core.
    Furthermore, the $k$-core is connected.
    Hence, it is not possible for the two cores at $k+1$ to be disconnected.
\end{proof}
\begin{lemma}\label{lem:st-size}
    Let $G=(V,E)$ be a graph with $n = \abs{V}$.
    The number of nodes in the shell tree is at most $n+1$ and edges is at most $n$.
\end{lemma}
\begin{proof}
    By Lemma~\ref{lem:compress}, each node in the tree (besides the root)
    must have at least one vertex. As there are at most $n$ vertices, the
    size of the tree is at most $n+1$.  By Lemma~\ref{lem:st-tree}, we
    know it is a tree, and so with at most $n+1$ nodes it has at most $n$
    edges.
\end{proof}

\paragraph{Queries on \STI/}

The three queries, $\K/(u)$, $\C/(u, k)$, and $\H/$ are returned as follows.
\begin{itemize}
    \item $\K/(u)$ follows the map $M[u]$ to the shell tree node $n$, and
        then returns the $k$ value for $n$.
    \item $\C/(u,k)$ runs a tree traversal that stays above the level $k$
    \item $\H/$ returns the tree nodes and attributes directly.
\end{itemize}

\paragraph{Efficiency of \STI/}
We next address the efficiency of queries on \STI/.

\begin{theorem}
    $\C/(u,k)$ queries on \STI/ run in $O(\abs{\C/(u,k)})$ and correctly return the $k$-core.
\end{theorem}
\begin{proof}
    First, we show correctness.  Let $C^*$ be the core for $\C/(u,k)$,
    that is $C^*$ is a $k$-core and $u \in C^*$. The traversal will cover
    all vertices in the subtree containing $u$ at level $k$ and higher.  By
    Lemma~\ref{lem:laminar} we know all denser cores are fully contained
    in the desired $k$-core. By Lemma~\ref{lem:compress}, we know that any
    split will occur in an explicit tree node with vertices in the
    resulting shell. So, this split will be captured by the tree
    traversal.  As such, all vertices in the tree nodes traversed with values $k$ or more 
    exactly form the $k$-core.

    Let down represent higher $k$ values in the tree.
    Next, we show efficiency. Every downward link in the subtree needs to
    be fully explored, and there are no nodes with overlapping vertices in
    the tree. Once a downward traversal occurs, there is no need to check
    parents. When traversing upwards, all children except the previous one
    will be explored downwards.  In each case every node is visited
    exactly once and all of its associated vertices are enumerated once and are part of the returned core.

    As \STI/ is a tree, whether to traverse to the parent can be decided based on whether the parents' value is lower than $k$.
    This will result in one additional operation.
    As such, the runtime is $O(\abs{\C/(u,k)})$ and
    efficient.
\end{proof}
\begin{theorem}
    The \STI/ takes $O(n)$ space.
\end{theorem}
\begin{proof}
    The \STI/ consists of a map of size $n$ between vertices and tree nodes,
    along with the shell tree itself. By Lemma~\ref{lem:st-size}, the tree
    has at most $n+1$ nodes and $n$ tree edges. Each tree node may have
    vertices, but there are no redundant vertices. So, the size is
    $O(n+n+1+n+n) = O(n)$.
\end{proof}

The shell tree itself contains the hierarchy of
cores and shells, and so returning \STI/ efficiently resolves $\H/$
queries.

\begin{figure}[t]
    \centering
    \includegraphics{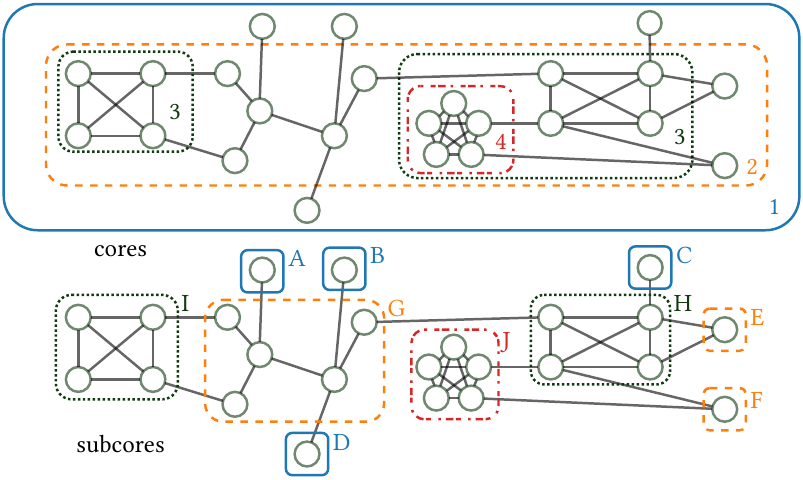}
    \caption{\label{fig:cf}
    An example graph along with its cores (top) and subcores (bottom).
    Note that a core may consist of multiple subcores, and subcores
    are disjoint.
    }
\end{figure}

\section{Computing the \STI/} \label{sec:computing-sti}

Computing (and maintaining) the \STI/ hinges on building (and maintaining)
the shell tree.  We propose a \emph{subcore directed acyclic graph}, that
provides the link between core decompositions and the shell tree.  In this
section we describe how to compute the \STI/ from scratch using the
subcore DAG.

This problem is broken into three parts:
computing coreness values, subcore DAG, and the shell tree.

\paragraph{Computing Coreness Values}

Computing coreness values has been well studied on graphs~\cite{matula1983smallest,dhulipala2017julienne}.  The most direct
approach, known as peeling, starts by keeping an array of vertex degrees. It then moves up
through coreness values, removing vertices with insufficient degree and
recording when they are removed. This is efficient, running in
$O(n+m)$, when using buckets~\cite{matula1983smallest}.
We refer the reader to \cite{malliaros2020core} for a survey.

\paragraph{Computing the Subcore DAG}

Next, we introduce the \emph{subcore directed acyclic graph (DAG)}, which is used to bridge
between coreness values and cores.

\begin{definition}
    Let $G$ be a graph.  A \emph{subcore} is a subgraph $C$ such
    that (1) $C$ is maximal (2) $\forall v \in C$, $\kappa[v] = k$ for
    some $k \in \mathbb{N}$ and (3) $C$ is connected.
\end{definition}
Subcores were introduced in~\cite{sariyuce2013streaming} to limit the
region that may have coreness values change on graph changes.
Figure~\ref{fig:cf} shows an example graph with cores and subcores.

\begin{figure}
    \centering
    \includegraphics{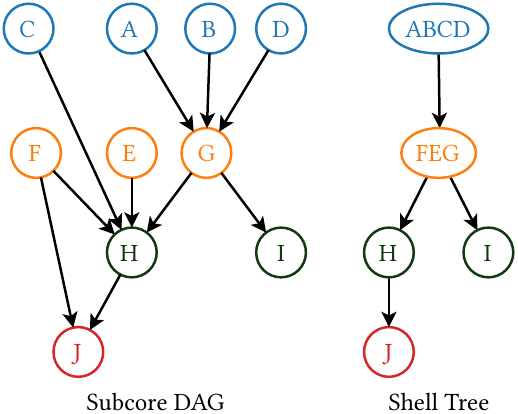}
    \caption{\label{fig:cfd}
    The corresponding subcore DAG and shell tree from the example
    graph in Figure~\ref{fig:cf}.
    }
\end{figure}

\begin{obs}\label{obs:cf-disjoint}
    Subcores are disjoint, by maximality of cores and property (2), and so
    the number of subcores is bound by $n$.
\end{obs}

After breaking cores up into subcores, the glue to link them back
together is saved as a \emph{subcore DAG}.
The subcore DAG is built with a directed edge from every
lower $k$ subcore to a strictly higher $k$ subcore that it is \emph{directly connected} to.
The subcore DAG from Figure~\ref{fig:cf} and its shell tree is shown in
Figure~\ref{fig:cfd}.

\begin{lemma}\label{lem:cfd-size}
    The subcore DAG size is bound by $G$.
\end{lemma}
\begin{proof}
    Each vertex in the subcore DAG corresponds to a connected
    subgraph in the graph, and every edge in the DAG is a directed edge
    that results from contracting all vertices in each subcore.
    Contraction only removes edges and vertices, and no new edges or vertices are added.
\end{proof}
\begin{obs}
    The subcore DAG is not a tree. Consider a $3$-clique and
    a $4$-clique, connected via an edge, and both connected to another
    vertex.  This forms a directed triangle in the DAG.
\end{obs}

\begin{algorithm}[tb]
    \KwIn{graph $G=(V,E)$, $\kappa$}
    $C \gets \emptyset$; $D \gets \emptyset$ \tcp{DAG vertices and edges}
    $L \gets [v : v \in V]$ \tcp{Labels}
    \kc{Compute the subcores}
    \For{$v \in V$}{
        \lIf{$L[v] \neq v$}{continue}
        $C \gets C \cup \{v\}$ \;
        \kc{Perform a BFS that stays within $\kappa$ levels from $v$}
        $Q \gets \mathrm{Queue}()$; $Q\mathrm{.push}(v)$ \label{alg:cfd:bfs-s} \;
        \While{$Q \neq \emptyset$}{
            $n \gets Q\mathrm{.pop}()$ \;
            \For{$w \in \Gamma(n): L[w] \neq v \land \kappa[w] = \kappa[v]$}{
                $Q\mathrm{.push}(w)$\;
                $L[w] = v$ \label{alg:cfd:bfs-e}\;
            }
        }
    }
    \kc{Produce the DAG edges}
    \For{$v \in V$ \label{alg:cfd:p-s}}{
        \For{$n \in \Gamma(v)$ where $L[v] \neq L[n]$}{
            $D \gets D \cup \{\tup{L[v], L[n]}\}$ \label{alg:cfd:p-e} \;
        }
    }
    \Return{DAG=$(C,D)$}
    \caption{\label{alg:cfd}
    Building the subcore DAG.
    }
\end{algorithm}

The process of building the subcore DAG is shown in
Algorithm~\ref{alg:cfd}.
This algorithm performs a breadth-fist search (BFS) for each vertex.
The search is constrained to stay within a $\kappa$ level, and DAG edges
are emitted on graph edges that leave $\kappa$ levels.
Efficient connected components algorithms,
e.g.,~\cite{shun2014simple}, could be used instead.

\begin{lemma}
    Algorithm~\ref{alg:cfd} runs in $O(n+m)$.
\end{lemma}
\begin{proof}
    From lines~\ref{alg:cfd:bfs-s}--\ref{alg:cfd:bfs-e}, inside the
    internal BFS, each vertex will be visited once. Inside, each
    edge will be visited once.  Finally, the entire BFS will only start from
    unvisited vertices.

    For lines~\ref{alg:cfd:p-s}--\ref{alg:cfd:p-e}, each vertex and edge
    will again be visited, resulting in $O(n+m)$ work.
\end{proof}

\subsection{Building the Shell Tree}

Given a subcore DAG and $\kappa$ values, we can compute the shell tree.
Our algorithm starts with the DAG and modifies it as it moves from the
sinks upwards (towards lower $k$ values), using a max-heap.
Each processed vertex: 1) identifies neighbors that are
at its $\kappa$ level, and merges itself with them;
2) sets a single node that is an in-neighbor with
the closest $\kappa$ value as the tree parent; and 3) moves all other in-edges to the
identified parent, ensure it becomes a tree.  The details are presented in
Algorithm~\ref{alg:st}.

\begin{algorithm}[tb]
    \KwIn{DAG=$(C,D)$, $\kappa$}
    $T=(N,E) \gets$ DAG \;
    $S \gets \emptyset$ \;
    $H \gets \mathrm{Heap}()$ \tcp{Empty Heap}
    \For{sink $s \in N$}{
        $H\mathrm{.push}(\kappa[s], s)$\;
    }
    \While{$H \neq \empty$}{
        $v \gets H\mathrm{.pop}()$ \;
        \lIf{$v \in S$}{{\bf continue}}
        $S \gets S \cup \{v\}$ \;
        \kc{Merge with neighbors at same level}
        \While{$\exists n \in \Gamma(v) : \kappa[n] = \kappa[v]$}{
            $\mathsf{Merge}(v, n)$ \;
            $S \gets S \cup \{n\}$ \;
        }
        \kc{Move all remaining and new in neighbors}
        $t \gets \argmax_{n \in \Gamma^{\mathrm{in}}(v)} \kappa[n]$ \;
        \For{$n \in \Gamma^{\mathrm{in}}(v)$}{
            \lIf{$n \neq t$}{$\mathsf{MoveEdge}(\tup{n,v} \to \tup{n,t})$}
            $H\mathrm{.push}(\kappa[n], n)$ \;
        }
    }
    \Return{$T$}
    \caption{\label{alg:st}
        Constructing the shell tree.
    }
\end{algorithm}

\begin{lemma}
    Algorithm~\ref{alg:st} correctly builds the shell tree.
\end{lemma}
\begin{proof}
    We argue that after running Algorithm~\ref{alg:st}, each node will exactly contain the shell.
    First, a node needs to contain all connected subcore DAG nodes at the given $\kappa$ value.
    Second, it cannot have additional nodes merged with it.
    We argue correctness via induction on $\kappa$.
    At the highest $\kappa$ level, by the DAG properties, we know the tree nodes connected to the sink are shells and valid.
    Now, consider a tree node with $\kappa$ and assume nodes at $\kappa' > \kappa$ are valid.
    The node is formed by merging DAG nodes at the same level, which are all connected.
    Any connectivity that is not at level $\kappa$ will be preserved by moving edges to the node's parent.
    By Lemma~\ref{lem:laminar}, we know that any DAG
    neighbors that it is connected to will also be connected to the
    parent, and so the new tree node is valid.
\end{proof}
\begin{lemma}
    Algorithm~\ref{alg:st} runs in $O(\rho(n+m)\log n)$.
\end{lemma}
\begin{proof}
    The heap processes each vertex once, and each vertex can
    potentially have all edges attached, resulting in $O(n+m)$ per
    iteration.
    However, edges may be carried upwards, and in the worst case
    all edges except one are carried upwards resulting in a factor of
    $\rho$. The log factor comes from the heap use.
\end{proof}

\section{Maintaining the \STI/} \label{sec:maint-sti}

In this section, we show how to maintain the \STI/ on a graph stream.  The
objective is to develop a batch dynamic algorithm $\mathcal{A}_\Delta$
that will output the shell tree \STI/, while having a small internal state
$\mathcal{A}_\Delta^\mathrm{s}$ and a quick runtime with low variability.

\subsection{Maintaining Coreness}

We refer the reader to
\cite{zhang2017fast,sariyuce2013streaming,li2013efficient,Gabert21-ParSocial} for algorithms
to maintain $\kappa$.
These approaches (and similarly \STI/) extend to trusses~\cite{cohen2008trusses} and other nuclei~\cite{sariyuce2015finding} by use of a hypergraph~\cite{Gabert21-WSDM}.
For our experiments we implemented and use
\Order/~\cite{zhang2017fast}, the state-of-the-art decomposition
maintenance algorithm.

For notational convenience, consider a time $t$.
Let $G^-$ denote $G^{(t)}$ and $G^+$ denote $G^{(t+\Delta)}$.
Let $\kappa^{-}$ denote the $\kappa$ values in $G^-$ and $\kappa^{+}$
denote $\kappa$ values in $G^+$.

We take advantage of the following crucial property of
coreness values on graphs: the subcore theorem.
\begin{theorem}[\cite{sariyuce2013streaming}]\label{thm:subcore}
    Let $\{u,v\}$ be an edge change.
    Suppose $\kappa_{G^-}[u] \leq \kappa_{G^-}[v]$.
    Then, only vertices in the subcore containing $u$ may have $\kappa$
    values change in $G^+$, and they may only change by 1 (increase by $1$ for
    insertion, decrease by $1$ for deletion.)
\end{theorem}

\subsection{Single Edge Maintenance Algorithm} \label{sec:se}

\begin{algorithm}
    \KwIn{graph $G=(V,E)$, $e=\{u,v\}$, $\kappa^-$, $\kappa^+$, \STI/ $=(M, T)$}
    \lIf{$\kappa^-[u] > \kappa^-[v]$}{swap $u$, $v$}
    $K \gets M[u]$ \tcp*{find the tree node for $u$}
    $S \gets \{ w \in V : \kappa^-[w] \neq \kappa^+\}$ \;
    \If{$M[u]\mathrm{.vertices} = S$}{
        \kc{The entire shell moves as one subcore}
        \For{$c \in K\mathrm{.children}$}{
            \lIf{$c\mathrm{.k} = k+1$}{$\mathsf{Merge}(K, c)$}
        }
        $K\mathrm{.k} \gets k+1$ \;
        \Return{$T$}
    }
    \kc{We need to merge or create a new sink}
    $K\mathrm{.vertices} \gets K\mathrm{.vertices} \setminus S$ \;
    $X \gets \tup{K, k+1, S}$ \tcp*{new tree node with parent $K$, level $k+1$, vertices $S$}
    \For{$w \in S$}{
        \For{$n \in \Gamma_{G^-}(w) \setminus S$}{
            \lIf{$\kappa^+[n] \geq k+1$}{
                $\mathsf{MergeOrConnect}(X, M[n])$
            }
        }
    }
    \kc{Merge the path with $v$}
    $c \gets M[v]$, $l \gets \mathrm{SINK}$ \;
    \While{$\kappa[c] \geq \kappa^+[u]$}{
        $l \gets c$; $c \gets c\mathrm{.parent}$ \;
    }
    $\mathsf{MergePaths}(X, c)$ \;
    \Return{$T$}
    \caption{\label{alg:maint-inc}
    \SE/ (incremental case).
    }
\end{algorithm}

\begin{algorithm}
    \KwIn{\STI/ $= (M,T)$, $U$, $V$}
    \lIf{$U = V$}{\Return{}}
    \lIf{$\kappa[U] > \kappa[V]$}{swap $U$, $V$}
    $c \gets V$; $l \gets \mathrm{SINK}$ \;
    \While{$\kappa[c] \geq \kappa[U]$}{
        $l \gets c$; $c \gets c\mathrm{.parent}$ \;
    }
    \If{$\kappa[U] = \kappa[c]$}{$\mathsf{Merge}(U, c)$\;
        \Return{$\mathsf{MergePaths}(c, U\mathrm{.parent})$}}
    \Else{$\mathsf{MakeChild}(U, c)$\;
        \Return{$\mathsf{MergePaths}(c, U)$}}
    \caption{\label{alg:merge-paths}
    $\mathsf{MergePaths}$, which merges two paths starting from tree nodes
    $U$ and $V$ until the root.
    }
\end{algorithm}

\begin{figure}
    \centering
    \includegraphics{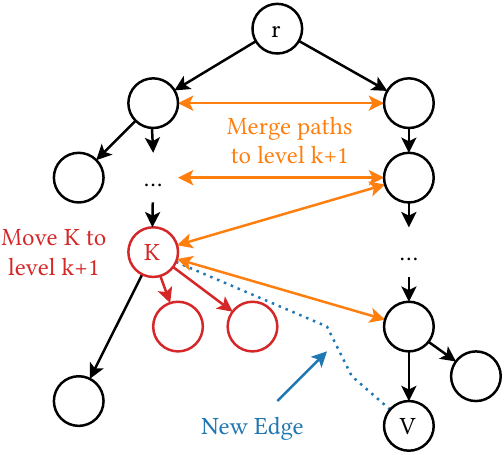}
    \caption{\label{fig:inc}
    The incremental algorithm process. First, the tree node corresponding
    to the smaller $\kappa$ level vertex, $K$, is processed.
    Next, the paths to $K$ and to the tree node
    being connected are merged to level $k+1$.
    }
\end{figure}

The main idea for maintaining the \STI/ edge-by-edge is to first break
apart any core
or shell that was increased and then repair the tree by merging together
the paths from the endpoints. For deletions,
a map is made that determines where, after a core is split, it could
return to in the tree.
Then, the path from the core to the root
is traversed and any potential split is determined.
Our algorithm shares many similarities to the community search algorithm of~\cite{fang2017effective}.
Our algorithm addresses cores instead of the more general community search problem on attributed graphs.
Specifically, it does not need to support queries involving subsets of vertices.
We refer to this approach as \SE/.
We describe insertions in detail---deletions are similar but split nodes~\cite{fang2017effective}.

Let $K$ be the tree node that has a \emph{lower $\kappa$ value} given an edge insertion.
We first check if all of $K$'s vertices leave.
If so, we move $K$ down and merge its children with connected subcores.
Next, we iterate through the moved
vertices and identify if they are connected to a shell tree node at
level $k+1$.
If so, we merge those shell tree nodes together. If not, we
create a new tree node for the moved vertices.
Then, we walk up the tree from both endpoints
and, starting at level $k+1$, begin merging all visited vertices.
The algorithm is presented in Algorithm~\ref{alg:maint-inc}, with merge paths presented in Algorithm~\ref{alg:merge-paths}.
A visual depiction is given in Figure~\ref{fig:inc}.

\begin{lemma}
    The runtime for Algorithm~\ref{alg:maint-inc} is $O(\abs{\Gamma(S)}
    + \rho n)$, where $S$ is the subcore that increases $\kappa$.
\end{lemma}
\begin{proof}
    In the first part, the modified subcore and all of its immediate
    neighbors are accessed, resulting in $O(\Gamma(S))$ work. After that,
    in the worst case, the height of the tree will be accessed to find the
    closest neighbor to merge in, resulting in $O(\rho n)$ work.
\end{proof}

\begin{figure}
    \centering
    \includegraphics{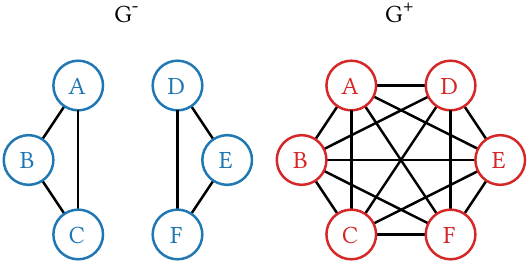}
    \caption{\label{fig:multins}
    An example graph before a batch of insertions ($G^{-}$) and after ($G^{+}$).
    The coreness moves from $\kappa=2$ to $\kappa=5$ for each vertex.
    }
\end{figure}

\begin{figure*}
    \centering
    \includegraphics{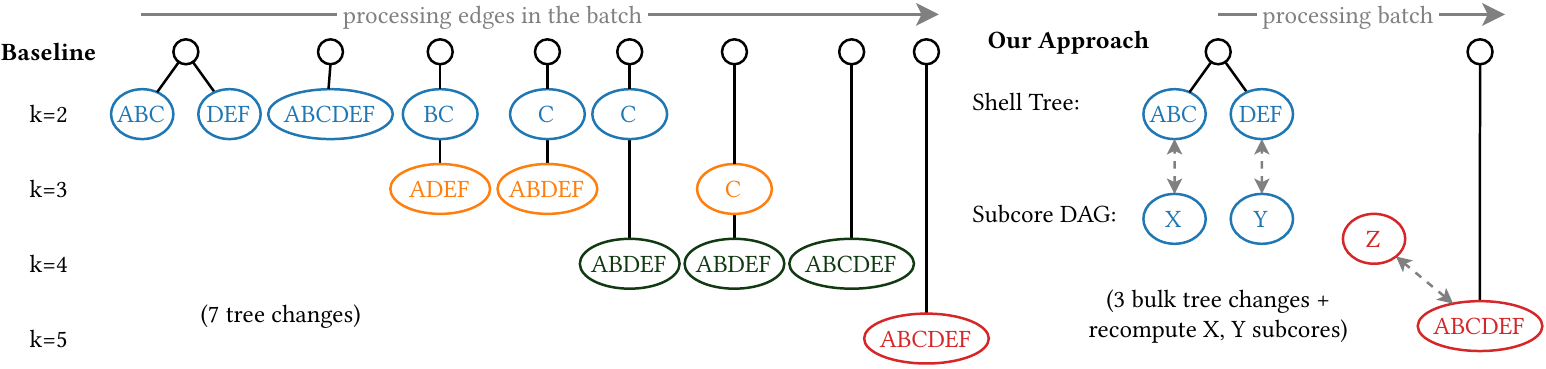}
    \caption{\label{fig:multiins-baseline}
    Following the example in Figure~\ref{fig:multins}, we show the tree changes processing with \SE/ compared with our batch approach.
    The cost is an increase in memory to store the subcore DAG and unnecessary work if a modified subcore does not significantly change.
    }
\end{figure*}

\subsection{Batch Maintenance}

We now present our batch maintenance algorithm.
First, we present the opportunity for reducing work by providing an example.
In Figure~\ref{fig:multins}, we show the graph before and after the batch.

The idea is to \emph{keep the subcore DAG in memory} and use it to update the subcore tree.
This can naturally be combined with \SE/ to provide a hybrid approach, moving between the two based on a batch size.
We maintain an additional pointer between every node in the tree and every node in the subcore DAG.
There are two main parts to maintaining the subcore tree in the subcore batch algorithm.
First, we maintain the subcore DAG by iterating over changed vertices and recomputing any subcore changes, creating and merging subcores (locally) as appropriate.
Second, we need to maintain the \STI/ given the DAG changes.
To do this we begin by making all of the DAG changes propagate forward to the tree.
Any deleted DAG node results in deleting the reference from the subcore tree, any newly empty tree nodes are deleted, and any new DAG nodes and their connections are added to the tree.
The tree is now no longer a DAG.
We then run the heap-based Algorithm~\ref{alg:st} to finish turning the modified structure back into a tree.
During this process we maintain the reverse vertex maps.
Unlike \SE/, our batch approach naturally covers deletions identically to insertions and both insertions and deletions can be mixed inside of batches.
This is due to handling both endpoints of an edge change, instead of only the endpoint with a lower $\kappa$ value at some point in time.
The approach is shown in Algorithm~\ref{alg:batch}.
Following the example in Figure~\ref{fig:multins}, we show the saved work between \SE/ and \BA/ in Figure~\ref{fig:multiins-baseline} (next page).

\begin{algorithm}
    \KwIn{\STI/ $= (M,T)$, DAG $D$, batch $B$}
    $C \gets \{ v : v \in e \in B \}$;
    $K \gets \emptyset$ \;
    $I \gets \emptyset$ \tcp*{Visited set}
    \For{$v \in C$}{
        \lIf{$v \in I$}{continue}
        $I \gets I \cup \{v\}$\;
        $Q \gets \mathrm{Queue}$; $Q\mathrm{.push}(v)$ \tcp*{Change queue}
        \While{$Q \neq \emptyset$}{
            $q \gets Q\mathrm{.pop}()$\;
            $n_d,n_T \gets L[q]$ \tcp*{DAG/Tree node of $q$}
            $K \gets K \cup \{ n_D \}$\;
            $n_d' \gets $ new DAG node \;
            assign $q$ to $n_d'$ in $D$ and $M, T$ \;
            $S \gets \mathrm{Queue}$; $S\mathrm{.push}(q)$ \tcp*{Subcore queue}
            \While{$S \neq \emptyset$}{
                $n \gets S\mathrm{.pop}()$ \;
                \tcc{Check if $n$ is in the subcore}
                \If{$\kappa^+[n] \neq \kappa^+[q]$ }{
                    \tcc{If $n$ changed, process it separately}
                    \If{$n \not\in I$ and $\kappa^-[n] \neq \kappa^+[n]$}{
                        $I \gets I \cup \{n\}$ \;
                        $Q\mathrm{.push}(n)$ \;
                    }
                    continue \;
                }
                \If{$n \not\in I$}{
                    $I \gets I \cup \{n\}$ \;
                    $Q\mathrm{.push}(n)$ \;
                }
                assign $n$ to $n_d'$ in $D$ and $M,T$ \;
            }
        }
    }
    remove newly isolated nodes in $D$ \;
    copy DAG edges from DAG nodes in $K$ to $T$ \;
    remove newly empty tree nodes in $T$ \;
    run Algorithm~\ref{alg:st} \;
    \caption{\label{alg:batch}
    The \BA/ algorithm.
    }
\end{algorithm}

Our runtime is the cost of Algorithm~\ref{alg:st} plus the cost of a BFS over each modified subcore.
Correctness follows from Algorithm~\ref{alg:st} as we maintain the built data structures and operations.
In the worst case this can be the runtime of Algorithm~\ref{alg:st}.
However, note that the BFS on subcores is limited to modified subcores.
As such, empirically we run faster than re-computing from scratch, as shown in the following Section~\ref{sec:experiments}.

\section{Empirical Analysis}\label{sec:experiments}

In this section we perform an experimental evaluation of our approach to
demonstrate that it is able to provide core queries on rapidly changing
real-world graphs.

\paragraph{Environment}
We implemented our algorithm in C++ and compiled with GCC 10.2.0 at
\texttt{O3}.
We ran on Intel Xeon E5-2683 v4 CPUs at 2.1 GHz with 256 GB of RAM and CentOS 7.
To perform coreness maintenance, we implemented \Order/~\cite{zhang2017fast}.
Any coreness maintenance approach can be used in its place.
We include all memory allocation costs in our runtimes.
We use a hash map of vectors to store the graph, and store both in- and
out-edges.
We ran five trials for each experiment and show the results from all
trials.

\paragraph{Baseline}
As our baseline, we implemented the non-batch maintenance approach
from~\cite{fang2017effective}, which we ported to the case of computing cores on graphs (see Section~\ref{sec:se}). We refer to this as \SE/.
When operating on a batch, \SE/ runs independently for each
edge change. Insertions and deletions can therefore easily be mixed.
We only show results with insertions as they are the harder case~\cite{fang2017effective} and there are few known benchmark datasets with frequent deletions.

\paragraph{Datasets}
The graphs that we evaluate with are benchmark graphs that are
representative of real-world graphs from a variety of domains and with
different properties.  We downloaded them from SNAP~\cite{snapnets}
(excluding Ar-2005, downloaded from~\cite{BoVWFI}).  The graphs we use are given in
Table~\ref{tab:datasets}.  We cleaned the data by removing self loops and
duplicates edges and treated graphs as undirected.
We randomized the edge order, simulating a graph stream, and performed our experiments by first removing random edges and next inserting them.

\begin{table}
    \caption{\label{tab:datasets}
    Graphs used with $n$, $m$ in millions.
    }
    \centering
    \begin{tabular}{lrrr} \toprule
        Name & $n$,\ \ \ $m$ & DAG $n$, \ \ $m$ & $\abs{T}$ \\ \midrule
        Ar-2005~\cite{BoVWFI,BRSLLP} & 22, 640 & 12, \ 47 & 28 K\\
        Orkut~\cite{yang2015defining} & 3, 117 & 1, \ 22 & 254 \ \ \ \\
        LiveJ~\cite{yang2015defining} & 4, \ 35 & 2, \ 12 & 2 K \\
        Pokec~\cite{takac2012data} & 2, \ 22 & 1, \ \ 5 & 54 \ \ \ \\
        Patents~\cite{leskovec2005graphs} & 4, \ 17 & 2, \ \ 4 & 4 K \\
        BerkStan~\cite{leskovec2009community} & 0.7, \ \ 7 & 0.2, 0.8 & 2 K \\
        Google~\cite{leskovec2009community} & 1, \ \ 4 & 0.4, 1.2 & 5 K \\
        YouTube~\cite{yang2015defining} & 1, \ \ 3 & 1, 2.5 & 140 \ \ \  \\
        \bottomrule
    \end{tabular}
\end{table}

\begin{figure}
    \centering
    \includegraphics{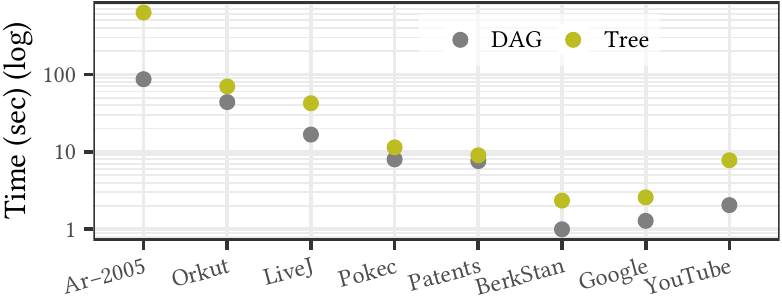}
    \caption{\label{fig:index}
    The \STI/ construction time, broken down into DAG construction and Tree construction.
    }
\end{figure}

\begin{figure}
    \centering
    \includegraphics{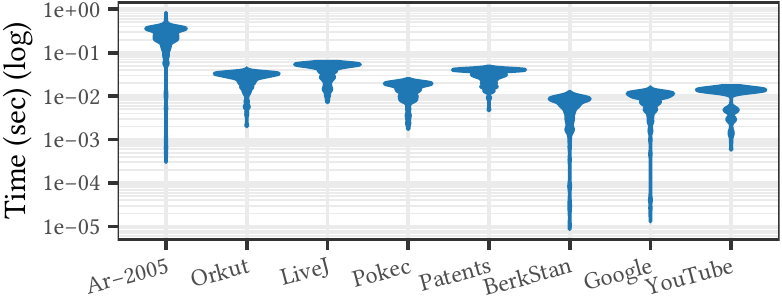}
    \caption{\label{fig:qC}
    The runtime to return \C/ queries.
    On all graphs, the runtimes are low enough for interactive use.
    }
\end{figure}

\begin{figure}
    \centering
    \includegraphics{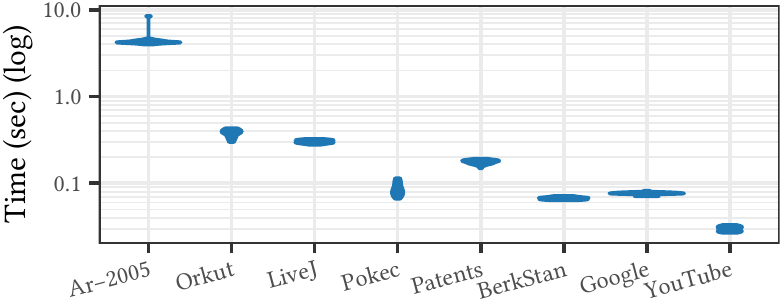}
    \caption{\label{fig:qH}
    The runtime to return \H/ queries.
    }
\end{figure}

\begin{figure*}
    \centering
    \includegraphics{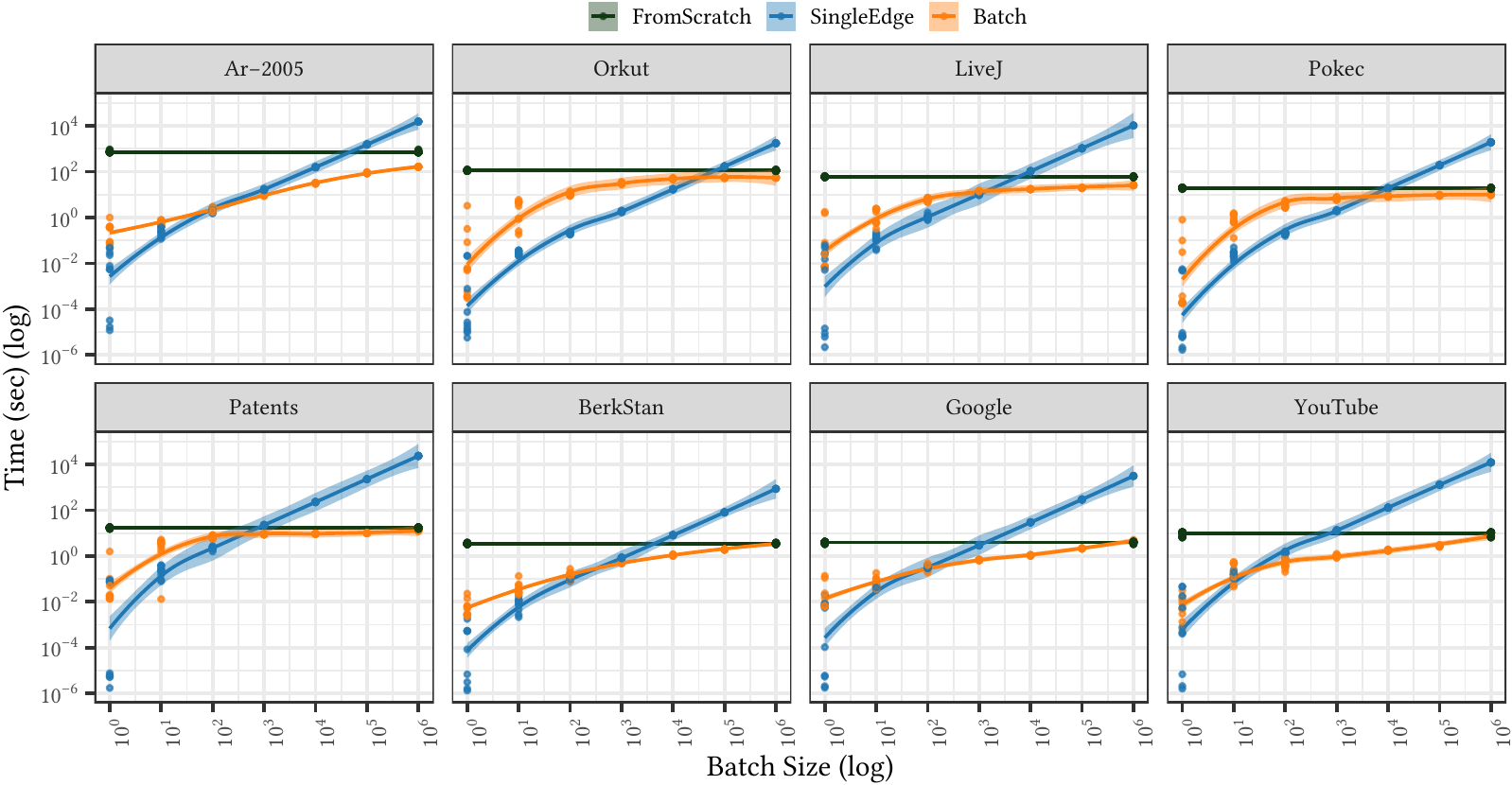}
    \caption{\label{fig:varybatch}
    Varying the batch size and running \BA/, \SE/, and \FS/.
    The batch algorithm is orders of magnitude faster than \SE/ for batches above $10^5$ and remains below re-computing from scratch up to $10^6$.
    The data points and LOESS smoothing lines with 95\% confidence intervals are shown.
    }
\end{figure*}

\paragraph{Experiments}

Our main experimental goal is to evaluate the real-world feasibility of
our approach on modern graphs and systems with highly variable and large batch
sizes.

First, we show the index construction time for \BA/.
The results are shown in Figure~\ref{fig:index}.
In all cases building the tree is more expensive than building the DAG.
The overall runtime reinforces the need for dynamic algorithms as for
large graphs, such as Orkut, the DAG construction takes around 90 seconds
and the tree construction takes around 330 seconds.

Next, we want to show that \STI/ is a useful index for cores.
We report the query times for \C/ in Figure~\ref{fig:qC} and \H/ in Figure~\ref{fig:qH} on \STI/.
For \C/, we performed queries from 1000 randomly sampled vertices with uniformly random $k$-values such that the vertex is in a $k$-core.
For all graphs, all cores are returned in under one second with many in the tens of milliseconds.
Given that our query is efficient the runtime largely consists of copying memory.
The denser the core the faster the return tends to be, as there are fewer vertices to copy out.
In many cases, the runtimes are fast enough to be used for interactive applications, e.g., in web page content.
For \H/, we report the time to build and return the full hierarchy, including each node at each level.
This is under 10 seconds for all graphs, showing that full hierarchies can be used for interactive time applications.

Finally, we maintained cores for 100 batches of different batch sizes for each graph.
The results are shown in Figure~\ref{fig:varybatch}.
In all cases, when batch sizes are large \BA/ remains below both \FS/ and
\SE/.
For a batch dynamic algorithm, we are looking for the region below re-computing from scratch and below single-edge algorithsm.
In some graphs, such as Pokec and Patents, it is not a large region, however in all graphs it exists and provides significant improvements.
Future work involves combining the DAG construction and maintenance with the direct tree maintenance to achieve an effective hybrid approach, achieving the lower of the all of the curves.
Note that these are log-log plots, and so even for Patents our batch approach is $2 \times$ faster than re-computing from scratch at batch sizes of one million.

\section{Conclusion} \label{sec:conclusion}

We focus on the important but overlooked problem of returning
\emph{cores}, as opposed to \emph{coreness} values.
We consider both core
queries, which return a $k$-core, and hierarchy queries, which return the full core hierarchy.
Our approach applies beyond $k$-cores to other arbitrary nuclei, such as trusses.

We develop algorithms around a tree-based index, the \STI/, that
is efficient and takes linear space in the number of graph vertices.
We provide an algorithm to construct the \STI/ using a new approach based on a subcore DAG.
We design and implement a batch maintenance algorithm for \STI/ that uses the same subcore DAG and can handle variable and high batch sizes.
We show that our approach is able to run faster than edge-by-edge approaches on rapidly changing graphs and can return cores and hierarchies fast enough for interactive use.

\begin{acks}
    This work was funded in part by the NSF under Grant CCF-1919021 and in part by the Laboratory Directed Research and Development program at Sandia National Laboratories.
    Sandia National Laboratories is a multimission laboratory managed and operated by National Technology \& Engineering Solutions of Sandia, LLC, a wholly owned subsidiary of Honeywell International Inc., for the U.S. Department of Energy's National Nuclear Security Administration under contract DE-NA0003525.
\end{acks}

\balance

\bibliographystyle{abbrv}

\end{document}